\documentclass{article}
\usepackage{amsmath,amsfonts,amsthm}

\usepackage{mathrsfs}

 \newtheorem{thm}{Theorem}[section]

 \newtheorem{cor}[thm]{Corollary}
 \newtheorem{lem}[thm]{Lemma}
 
 \theoremstyle{definition}
 \newtheorem{defn}[thm]{Definition}
 \newtheorem{rem}[thm]{Remark}

 \numberwithin{equation}{subsection}

 \newcommand{\Gd}{\Gamma_{\!0}}
 \newcommand{\Gn}{\Gamma_{\!1}}

 \newcommand{\Real}{\mathbb{R}}
 \newcommand{\Complex}{\mathbb{C}}

\newcommand{\Dom}{\mathcal{D}}
\newcommand{\Ran}{\mathcal{R}}
\newcommand{\Ker}{\mathrm{Ker}}

\newcommand{\bd}{{\beta_0}}
\newcommand{\bn}{{\beta_1}}

 \newcommand{\D}{{D}}
 
 \newcommand{\Disc}{\mathbb{D}}
 \newcommand{\Torus}{\mathbb{T}}
 \newcommand{\bD}{{\partial\D}}
 
 \newcommand{\G}{{\Pi}} 

\newcommand{\Ad}{A_0}

\begin{document}

\title{A Note on Operator-Theoretic Approach \\
 to Classic Boundary Value Problems \\
  for Harmonic and Analytic Functions \\
  in the Complex Plane Domain}

\author{Vladimir Ryzhov}




%
%

\date{August 30, 2009}


\maketitle

\begin{abstract}
 General
 spectral boundary value problems
 framework
 is utilized to restate
  Poincar\'e, Hilbert,  and Riemann problems
 for harmonic and analytic functions
 in the abstract operator-theoretic setting.
\end{abstract}


\section*{Introduction}

The last several years have witnessed increased
 interest revealed by the mathematical community
  to the abstract operator-theoretic methods
   in applications to spectral boundary value problems
    for differential
     operators and operator matrices.
 It is sufficient to point out
  numerous recently published
   works~\cite{AlpBe,AmP, AGMT,AGMST, BehrLan, BGW,
   BHMNW,BM,BMNW, GM1, GM2, GM3, GMZ1, GMZ2, Posil3, Ryz1, Ryz2, Ryz3}
    along with their extensive bibliographies in order
     to appreciate
      the potential and vitality of emerging concepts and approaches.
 The general theory has been successfully applied to
  boundary value problems
  for general elliptic partial differential operators
   of even order in bounded Lipschitz domains,
   for nonselfadjoint $(2\times 2)$-block operator matrices
   acting in $L^2(0,1)\times L^2(0,1)$
   known as Hain-L\"ust operators,
   for additive perturbations of
   multiplication operators
   and some other cases inspired by the theory
   of elliptic partial differential operators.
%


%
The presented paper is an attempt to
 embrace the study of
  boundary value problems of complex analysis
   by the general operator theoretic framework.
 We follow the line of reasoning
   developed in~\cite{Ryz1,Ryz2, Ryz3} and
    hope to demonstrate utility of the abstract technique
    in formulating classic problems of Poincar\'e,
     Hilbert, and Riemann  for harmonic and analytic
      functions in the bounded simply connected and
       sufficiently smooth
        domain of the complex plane.
 Keeping this goal in mind
  no attempt is made to report
   any function analytic results
    on solvability and
     properties of solutions of these problems.
 For the comprehensive treatment (at least in
  the classical settings)
   the interested reader is referred to the
    authoritative
     resources~\cite{Beg, Ga, Mikh, Mus},
      where all the details can be found.
%


%
The paper consists of two sections.
 After recollecting relevant definitions and statements
  from~\cite{Ryz1,Ryz2, Ryz3}
   we apply obtained results to the Laplace operator
    on the plane domain.
Then by appropriate choice of
 boundary conditions we arrive
  at the standard statements of three
   aforementioned problems of complex analysis.
%


As usual, $\Real$,  $\Complex$ are the sets of real
 and complex numbers.
For two separable Hilbert spaces $H_1$, $H_2$ and
 linear operator~$A$
 from $H_1$  to $H_2$ the notation $A : H_1\to H_2 $ means
 that $A$ is defined everywhere in $H_1$ and bounded.
 Domain, range, and kernel of $A$ are $\D(A)$, $\Ran(A)$,
 and $\ker (A)$, respectively.
 The writing~$A : f\mapsto g$ for $f \in \Dom(A)$
 is equivalent to $Af = g$.
 The symbol~$\rho(A)$ is used for the resolvent
  set of~$A$.
If  $A : H \to H$ and  $\lambda \in \Complex$, then
  the inclusion~$\lambda^{-1}\in \rho(A)$
 means that the operator~$I - \lambda A$ is boundedly invertible,
 i.~e.  the inverse~$(I - \lambda A)^{-1}$ exists
 and is bounded in $H$.
 When discussing function theoretic concepts,
  the Lebesgue measure is assumed.




\section{Spectral Boundary Value Problems}

\subsection{Spaces and operators}
  Let $H$ be a Hilbert space and $T : H \to H$ be a bounded
linear operator.
Assume $\ker(T) = \{0\}$ and denote $\Ad$ the left inverse of $T$
 so that
\[
 \Ad T f = f, \quad f \in H
\]
Note that~$\Ad$ with domain~$\Dom(\Ad) = \Ran(T)$ need not be
 bounded, closed or even densely defined.
Let~$E$ be another Hilbert space and
 $\G : E \to H $ be a linear mapping with $\Ker(\G) = \{0\}$
 satisfying condition
\[
 \Ran(T) \cap \Ran(\G) = \{0\}
\]
 It follows that the linear set $\Ran(T) + \Ran(\G)$ is the
 direct sum~$\Ran(T) \dot + \Ran(\G)$.
Introduce linear operator~$A$ in $H$ with the domain $\Dom(A)
 := \Ran(T)\dot + \Ran(\G)$ by
 \[
  A : T f + \G \varphi \mapsto f, \qquad f \in H,\, \varphi \in E
\]
Obviously,
\[
 \Ker(A) = \Ran(\G), \qquad \Ran(A) = H, \qquad \Ad = A|_{\Ran(T) }
\]
Since $\Ker(\G) = \{0\}$ there exists the left inverse $\gamma_0$
of~$\G$ such that $\Ker(\gamma_0) = \{0\}$ and
\[
 \gamma_0 \G \varphi = \varphi, \qquad \varphi \in E
\]
We extend the operator~$\gamma_0$ from its domain~$\Dom(\gamma_0) =
\Ran(\G)$ to the linear map~$\Gd$ defined on $\Dom(A)$ by
\[
\Gd : T f + \G \varphi \mapsto \varphi, \qquad f \in H,\, \varphi
\in E
\]
It is clear that
\[
\Ker(\Gd)  = \Ran(T), \qquad \Ran(\Gd) = E,\qquad \gamma_0 =
\Gd|_{\Ran(\G)}
\]

\subsection{Spectral boundary value problem}
The spectral boundary value problem for unknown~$u \in \Dom(A)$ is
defined by the system of two equations
 \begin{equation}\label{Eq:SBVP}
 \left\{
 \begin{aligned}
  &&(A - \lambda I) u = f \\
  && \Gd u = \varphi
 \end{aligned}
 \right.
 \end{equation}
where $f\in H$, $\varphi \in E$ and $\lambda \in \Complex$ is the
spectral parameter.
Since $T : H \to H$, the  bounded inverse $(I - \lambda T)^{-1}$
 exists for any $\lambda$ in a small neighborhood of $\lambda = 0$.
To justify the terminology we note that in the applications
 below the first equation~(\ref{Eq:SBVP})
 is realized as the ``main'' equation for the operator~$A$
  defined in a bounded domain of the complex plane,
   whereas equality~$\Gd u = \varphi$
    plays the role of boundary condition.
 The operator~$\Gd$ is interpreted as a ``boundary map''
  defined on $\Dom(A)$ with values
  in the ``boundary space''~$E$.


 \begin{lem}\label{Lem:KerA-lambda}
 Suppose $\lambda^{-1}
    \in \rho(T)$. Then
  \[
    \Ker(A -\lambda I) = \Ran (( I - \lambda T)^{-1}\G)
  \]
\end{lem}


\begin{proof}
 Let $u\in \Dom(A)$.
 Since $u\in \Dom(A) = \Ran(T) \dot + \Ran(\G)$
 there exist $f\in H$ and $\varphi \in E$ such that
 $u = T f + \G \varphi$.
 Then
 \[
 (A - \lambda I) u = (A - \lambda I)(T f + \G \varphi)
 = f - \lambda ( T f + \G \varphi) = (I - \lambda T) f- \lambda \G
 \varphi
 \]
 Assuming $(A - \lambda I) u = 0 $ and $\lambda^{-1} \in \rho(T)$
 we obtain
 $ f =  \lambda (I - \lambda T)^{-1}\G \varphi $.
 Substitution into~$u = T f + \G \varphi$ yields
 \[
   u =  \lambda T (I - \lambda T)^{-1}\G
   \varphi + \G \varphi = [I  + \lambda T (I - \lambda
   T)^{-1}]\G\varphi
   =
(I - \lambda T)^{-1}\G\varphi
 \]
To prove the inverse, put  $v = (I - \lambda T)^{-1}\G\varphi$ with
some $\varphi \in E$ and observe that
\[
 \begin{aligned}
 (A- \lambda I)(I - \lambda T)^{-1} & =
 (A - \lambda I) [I + \lambda T (I - \lambda T)^{-1}]
 \\
 & = A - \lambda I + \lambda  (A - \lambda I) T (I - \lambda T)^{-1}
  \\
  & = A - \lambda I +  \lambda (I - \lambda T) (I - \lambda T)^{-1}
  = A
 \end{aligned}
\]
Since $\Ker(A) = \Ran(\G)$,
\[
 (A - \lambda I) v = (A- \lambda I)(I - \lambda T)^{-1}\G\varphi
  =  A\G
  \varphi = 0
\]
which completes the proof.
\end{proof}

The following theorem  describes solutions of~(\ref{Eq:SBVP})
when $\lambda^{-1}\in \rho(T)$ .

\begin{thm}\label{Th:BVPTheorem1}
 If $\lambda^{-1} \in \rho (T)$, then
 the problem~(\ref{Eq:SBVP})
 is uniquely solvable for any
 $f \in H$, $\varphi \in E$
 with the solution
 \begin{equation}\label{Eq:Soultion}
  u_\lambda^{f,\varphi} = T(I - \lambda T)^{-1} f + (I - \lambda
  T)^{-1}\G \varphi
 \end{equation}
\end{thm}

\begin{proof}
  Uniqueness of solution follows from the standard arguments.
 Namely,
 if $u_1, u_2 \in \Dom(A)$ are two solutions, then for their
 difference~$u_0 = u_1 - u_2 = A f_0 + \G \varphi_0$ with some $f_0 \in
 H$, $\varphi_0 \in E$ we have
$ (A -\lambda I)u_0 = 0 $ and $\Gd u_0 =0$.
 Since $\Ker(\Gd) = \Ran(T)$ and $\Gd \G = I$, the second identity
 gives~$\varphi_0 =0$.
 Then the first identity yields~$0 = (A -\lambda I)T f_0 = (I - \lambda
 T)f_0$ and the equality $f_0 =0 $ follows from the assumption~$\lambda^{-1} \in
 \rho(T)$.

Let us prove the representation~(\ref{Eq:Soultion}).
Due to Lemma~\ref{Lem:KerA-lambda} the term $(I - \lambda
  T)^{-1}\G \varphi$ belongs to $\Ker(A - \lambda I)$.
Thus
   we have
\[
  (A -\lambda I) u_\lambda^{f,g}  =
 (A -\lambda I) T(I - \lambda T)^{-1} f
 = ( A -\lambda)(I - \lambda T)^{-1} T f  = AT f = f
\]
The condition~$\Gd u = 0$ is fulfilled for~$u_\lambda^{f,\varphi}$
  defined by~(\ref{Eq:Soultion}) due to obvious calculations
 \[
  \Gd u_\lambda^{f,\varphi} = \Gd  (I - \lambda
  T)^{-1}\G \varphi =
  \Gd [I + \lambda T   (I - \lambda
  T)^{-1}]\G \varphi = \Gd \G \varphi = \varphi
 \]
where we used equality $\Ker(\Gd) = \Ran(T)$.
The proof is compete.
\end{proof}


\subsection{M-operator} Let
 $\Lambda$ be a linear operator in $E$ defined
 on the domain~$\Dom(\Lambda)\subset E$.
Introduce the linear map~$\Gn$
 on $\Dom(\Gn) = \Ran(T) \dot + \G\Dom(\Lambda) \subset \Dom(A)$
  with the range  $\Ran(\Gn)\subset E$ by
 \[
   \Gn : T f + \G \varphi \mapsto \G^* f + \Lambda \varphi,
   \qquad f \in H, \varphi \in \Dom(\Lambda)
 \]
Obviously,
\begin{equation}\label{Eq:GnTandGnG}
  \Gn T = \G^*, \qquad \Gn \G = \Lambda
\end{equation}
Note that~$\Gn T$ is bounded as an adjoint to the bounded operator.
In applications below $\Gn$ is realized as the ``second
 boundary operator'' complementary to~$\Gd$.

 \begin{defn}
   The M-operator is an operator-function~$M(\lambda)$
   of the spectral parameter~$\lambda$
   defined on $\Dom(\Lambda)$ for $\lambda^{-1} \in \rho(T)$
  by the equality
\[
  M(\lambda) \Gd u_\lambda = \Gn u_\lambda, \qquad
  u_\lambda \in \Ker(A -\lambda I) \cap \Dom(\Gn)
\]
 \end{defn}
To check correctness of this definition assume~$u_\lambda \in \Ker(A
- \lambda I) \cap \Dom(\Gn)$ and $\lambda^{-1} \in \rho (T)$.
Then according to  Lemma~\ref{Lem:KerA-lambda} $u_\lambda = (I -
\lambda T)^{-1} \G \varphi$ where $\Gd u_\lambda = \varphi$
 with some $\varphi \in \Dom(\Lambda)$.
Therefore  $\Gd u_\lambda =0 $ means $\varphi =0$, which in turn
implies $u_\lambda =0 $ and $\Gn u_\lambda =0$.

\begin{thm}\label{Th:MOperator}
 For $\lambda^{-1} \in \rho(T)$
\[
  M(\lambda)  =
  \Gn (I - \lambda T)^{-1} \G =
  \Lambda + \lambda \G^* (I - \lambda T)^{-1} \G,
  \qquad  \Dom(M(\lambda)) =
  \Dom(\Lambda)
\]
\end{thm}

\begin{proof}
 By Lemma~\ref{Lem:KerA-lambda}
 any $u_\lambda\in \Ker(A -\lambda I)$
 has the form  $u_\lambda = (I - \lambda T)^{-1}\G \varphi$ with some
 $\varphi \in E$.
Assuming~$\varphi \in \Dom(\Lambda)$ we have~$u_\lambda \in \Dom(\Gn
)$ and
\[
 \begin{aligned}
 \Gn u_\lambda  &= \Gn (I - \lambda T)^{-1}\G \varphi =
 \Gn [I + \lambda T (I - \lambda
 T)^{-1}]\G \varphi \\
 &= [ \Gn \G + \lambda \Gn T (I - \lambda
 T)^{-1}\G]\varphi
 \end{aligned}
\]
The statement follows from equalities $\Gd u_\lambda = \varphi$,
$\Gn T = \G^*$, and $\Gn \G = \Lambda$.
\end{proof}

\begin{cor}
Values of the analytic operator-function~$M(\lambda) - M(0)$,
$\lambda^{-1} \in \rho(T)$
 are bounded operators in $E$.
\end{cor}


\subsection{Boundary conditions}
Let $\bd$, $\bn$ be two linear operators, $\Dom(\bd) \supset
\Dom(\Lambda)$ and $\bn : E \to E$.
 Consider spectral boundary value problem
 for unknown~$u \in
\Dom(\Gn)\subset \Dom(A)$
 defined by
\begin{equation}\label{Eq:GeneralSystem}
 \left\{
 \begin{aligned}
  &&(A - \lambda I) u = f \\
  && ( \bd\Gd + \bn \Gn) u = \varphi
 \end{aligned}
 \right.
\end{equation} where $f\in H$, $\varphi \in E$ and $\lambda \in
\Complex$ is the spectral parameter.

\begin{thm}\label{Thm:SolvabilityB0B1}
  Assume $\lambda^{-1} \in \rho(T)$ is such that
   the equation
\begin{equation}\label{Eq:AuxEquation1}
 \left[\bd + \bn M(\lambda)\right]\psi = g
\end{equation}
 with unknown $\psi \in E$
 is uniquely solvable for any~$g \in E$.
 Then the boundary value problem~(\ref{Eq:GeneralSystem})
  has unique solution~$u_\lambda^{f,\varphi}\in \Dom(A)$
  given by
 \begin{equation}\label{Eq:GeneralSystemSolution}
 u_\lambda^{f,\varphi} = T(I - \lambda T)^{-1} f + (I - \lambda
 T)^{-1} \G \Psi_\lambda^{f,\varphi}
  \end{equation}
where $ \Psi_\lambda^{f,\varphi} \in E$ solves
(\ref{Eq:AuxEquation1}) with
 \begin{equation}\label{Eq:AuxEquation2}
g = \varphi - \bn \G^* (I - \lambda T)^{-1}f
 \end{equation}
\end{thm}

\begin{rem}
Formally the left hand side of~(\ref{Eq:AuxEquation1})
 is meaningful only for~$\psi \in \Dom(M(z))= \Dom(\Lambda)$.
 However, the domain of $\bd + \bn M(z)$ can be wider than
  $\Dom(\Lambda)$, for example if the operator sum~$\bd+\bn \Lambda$
  is bounded.
Taking such possibilities into consideration  the general solution
 to~(\ref{Eq:AuxEquation1})
  is sought in the whole space~$E$.
\end{rem}


\begin{proof}
  Due  to Lemma~\ref{Lem:KerA-lambda} the
  second term in~(\ref{Eq:GeneralSystemSolution})
   belongs to $\Ker(A - \lambda I)$.
Therefore
\[
(A - \lambda I)u_\lambda^{f,\varphi}  = (A - \lambda I) T(I -
\lambda T)^{-1} f = f
\]
Thus the element~(\ref{Eq:GeneralSystemSolution}) solves
the first equation in~(\ref{Eq:GeneralSystem}).
Let us verify fulfillment of the second equation in~(\ref{Eq:GeneralSystem}).
To that end we need to calculate  $(\bd\Gd +
\bn\Gn)u_\lambda^{f,\varphi}$ where $u_\lambda^{f,\varphi}$ is
defined by~(\ref{Eq:GeneralSystemSolution}).
Assuming for the moment that~$\Psi_\lambda^{f,\varphi} \in
\Dom(\Lambda)$ so that $u_\lambda^{f,\varphi}\in\Dom(\Gn) $, we have
according to properties of~$\Gd$, $\Gn$ and
Theorem~\ref{Th:MOperator}
\[
 \begin{aligned}
 \Gd u_\lambda^{f,\varphi} & = \Gd [T(I - \lambda T)^{-1} f + (I - \lambda
 T)^{-1} \G \Psi_\lambda^{f,\varphi}] = \Psi_\lambda^{f,\varphi}
 \\
 \Gn u_\lambda^{f,\varphi} & = \Gn [T(I - \lambda T)^{-1} f + (I - \lambda
 T)^{-1} \G \Psi_\lambda^{f,\varphi}]
 = \G^* (I - \lambda T)^{-1} f + M(\lambda) \Psi_\lambda^{f,\varphi}
 \end{aligned}
\]
Since~$\Psi_\lambda^{f,\varphi}$ solves~(\ref{Eq:AuxEquation1}),
 (\ref{Eq:AuxEquation2}), we have
\[
 \begin{aligned}
(\bd \Gd + \bn \Gn )  u_\lambda^{f,\varphi} & =
 \bd \Psi_\lambda^{f,\varphi} + \bn
 [ \G^* (I - \lambda T)^{-1} f + M(\lambda) \Psi_\lambda^{f,\varphi}]
 \\
  & = (\bd + \bn M(\lambda)) \Psi_\lambda^{f,\varphi} +
 \bn \G^* (I - \lambda T)^{-1} f
 \\
 & = \varphi - \bn \G^* (I - \lambda T)^{-1}f  +  \bn \G^* (I - \lambda T)^{-1}
 f = \varphi
 \end{aligned}
\]
Now the condition~$u_\lambda^{f,\varphi} \in
 \Dom(\Gn)$ can be relaxed by treating the expression~$\bd \Gd + \bn \Gn$
 as an operator sum initially defined on $\Dom(\Gn)$
 and then extended to its maximal domain in $\Dom(A) \subset E$.
%


Calculations above show that
 $u_\lambda^{f,\varphi}$ solves  the
system~$(A - \lambda I)u =f$, $\Gd u = \Psi_\lambda^{f,\varphi}$.
According to the uniqueness part of Theorem~\ref{Th:BVPTheorem1}
this solution is unique if equalities $f = 0 $ and
$\Psi_\lambda^{f,\varphi} = 0$ imply $\varphi =0$.
In turn, this implication follows from the unique solvability of
(\ref{Eq:AuxEquation1}).
The proof is complete.
\end{proof}


\subsection{Operator node}
In this subsection we discuss connections of the spectral boundary
 value problems~(\ref{Eq:SBVP}), (\ref{Eq:GeneralSystem})
 to the theory of open systems
 thereby translating the setting of previous sections into
 alternative, in some sense more intuitive, terms.
 We refer the reader to the books~\cite{CoFr, CuZw, Part, Staff}
 for background information on the linear systems theory.


%
 The collection~$\{T, \G, \Lambda; H, E\}$
  of two Hilbert spaces and three operators
   introduced above
    defines the block operator matrix
     acting in the space~$H\oplus E$
 and often
 called \textit{the operator node}
 \begin{equation}\label{Eq:NodeM}
\mathfrak{M} =  \bigg(
\begin{array}{ccc}
T &  \G \\
\G^*  &   \Lambda
\end{array}
 \bigg)
 \end{equation}
 The node~$\mathfrak M$  is associated with an open
 system~$\widehat {\mathfrak M}$ defined as follows.
The state and the input-output
 spaces of the system~$\widehat {\mathfrak M}$
 are identified with  $H$, $E$
 respectively.
 The inner states of~$\widehat {\mathfrak M}$
  are realized as elements of~$H$
  and are governed by the equation~$(A - \lambda I)u = 0$.
Elements of $E$ represent external control and observation
 data sent to the input and read from the output
 of the system~$\widehat {\mathfrak M}$
  by the external control and observation processes.
 For $\lambda^{-1} \in \rho(T)$ and $\varphi \in E$ the control
 process is given as the input-state
 mapping~$\varphi \mapsto u_\lambda^\varphi = (I - \lambda T)^{-1}\G \varphi$.
The state-output mapping representing the observation process is
 defined as  $u_\lambda^\varphi \mapsto \Gn u_\lambda^\varphi$
  assuming $u_\lambda^\varphi \in \Dom(\Gn)$, or equivalently,
   $\varphi \in\Dom(\Lambda)$.
In this model
  the transfer
 function that maps
  inputs into outputs coincide with the
  M-operator~$M(\lambda) : \varphi \mapsto\Gn u_\lambda^\varphi $.
  The map~$\Lambda$
   is called the feedthrough operator.
The role of $\Lambda$ becomes clear if we note that for $\lambda
=0 $ the input-output mapping reduces to the correspondence~$\varphi
\mapsto \Lambda \varphi $ .
%


The condition~$(\bd\Gd + \bn \Gn) u = \varphi$
 can be interpreted as a description of
 the system obtained from~$\widehat {\mathfrak M}$
  by ``mixing'' its inputs and outputs into a new input
   defined by the operator sum~$\bd\Gd + \bn \Gn$.
 The second term represents
  a feedback procedure that sends the original output~$\Gn u_\lambda^\varphi$,
    modified along the way by the operator~$\bn$,
     back to the system's input.
 In a similar way,
   with a suitable choice of operators~$\alpha_0$, $\alpha_1$,
   the output can be redefined
   as the sum~$(\alpha_0\Gd + \alpha_1 \Gn)u_\lambda^\varphi$,
  where $(A - \lambda)u_\lambda^\varphi = 0 $ and $\Gd u_\lambda^\varphi =
  \varphi$ is the input of system~$\widehat{\mathfrak M}$.
Combination of these two ``mixing'' operations leads to the system
with the output $(\alpha_0\Gd + \alpha_1 \Gn)u_\lambda^\varphi$
where~$u_\lambda^\varphi \in \Ker(A - \lambda I)$ is the state
satisfying condition~$(\bd\Gd + \bn \Gn) u = \varphi$, and $\varphi$
is considered as the input.
 The resulting system~$\widehat{\mathfrak N}$
  is called the fractional linear transformation of
   $\widehat {\mathfrak M}$.
It is not difficult to see that the mapping
 \[
 N(\lambda) : (\bd + \bn
 M(\lambda)) \varphi \mapsto (\alpha_0 + \alpha_1 M(\lambda)) \varphi
 \]
 is the transfer function of $\widehat{\mathfrak N}$.
 Here~$\varphi \in \Dom(\Lambda)$ is regarded as a parameter.
 If $(\bd +\bn M(\lambda))$ is boundedly invertible, then~$N(\lambda)$
 can be written in the form of linear operator
\[
 N(\lambda) = (\alpha_0 + \alpha_1 M(\lambda)) (\bd + \bn
 M(\lambda))^{-1}
 \]
In general case when $(\bd +\bn M(\lambda))$
 is not invertible~$N(\lambda)$ is a multi-valued map,
 or in other terminology,
 a linear relation on the Hilbert
 space~$E \oplus E$.
 Trivial inputs satisfying~$(\bd + \bn
 M(\lambda))\varphi = 0$
 correspond to the inner states that
 always exist and
 produce non-trivial output regardless
 of the input applied to the system.

Expression for the feedthrough operator~$\Theta$
 of system~$\widehat{\mathfrak N}$
 is obtained by setting~$\lambda = 0$,
\[
 \Theta : (\bd + \bn
\Lambda) \varphi \mapsto (\alpha_0 + \alpha_1 \Lambda)\varphi
\]
Assuming $(\bd + \bn \Lambda)$ is invertible, $ \Theta  = (\alpha_0
+ \alpha_1 \Lambda) (\bd + \bn \Lambda)^{-1} $.
Existence of both factors as well as existence of their product
 here and in the formula for~$N(\lambda)$ above requires
  further justification, especially in cases where participating
   operators are unbounded.
 The detailed discussion of relevant issues
  in the setting of abstract boundary value problems
   can be found in~\cite{Ryz3}.
 A brief illustration of these concepts is given below
 for the case of Hilbert boundary value problem
  for analytic functions.

\section{Applications}

Let $\D \subset \Complex $
 be a bounded simply connected domain of the complex plane~$\Complex$ with
 smooth boundary~$\bD$.
 Let us define the main and boundary Hilbert spaces as
  $H = L^2(\D)$, $E = L^2(\bD)$.
It is well known that the inhomogeneous boundary value problem for
the Dirichlet Laplacian in $H$
\begin{equation}\label{Eq:Harmonic1}
 \Delta u = f,\qquad u|_\bD = 0
 \end{equation}
 is uniquely solvable for any~$f \in H$.
 Let $T : H \to H$ be the corresponding solution operator~$T : f \mapsto u$ acting
 in $L^2(\D)$.
 The range~$\Ran(T)$ consists of all functions from
 the Sobolev class~$W_2^2(D)$ vanishing on the boundary~\cite{Brow}.
 Therefore~$\Ran(T)$
 is dense in $L^2(\D)$.
Following the general schema,
 we define~$\G : L^2(\bD) \to L^2(\D)$ to be the solution
 operator for the problem
\begin{equation}\label{Eq:Harmonic2}
 \Delta u = 0,\qquad u|_\bD = \varphi
\end{equation}
where $\varphi \in L^2(\bD)$.
Clearly, $\G$ is the operator of harmonic continuation of functions
 defined on $\bD$ into the domain~$\D$.
It is an integral operator with the kernel
 expressed in terms of  Green's function of the domain~$\D$.
If $u^\varphi$ is a solution to~(\ref{Eq:Harmonic2}) corresponding
 to~$\varphi \in L^2(\bD)$, then the element~$\varphi$ is uniquely
 (in sense of $L^2(\bD)$)
  recovered from $u^\varphi$ by the boundary trace
   mapping~$\gamma_0 : u \mapsto u|_\bD$.
Thus $\gamma_0\G = I_E $.

The solution of homogeneous problem $\Delta u =0$ with
 condition~$u|_\bD = 0 $
  is trivial and therefore
   the equality~$\Ran(T) \cap \Ran(\G) = \{0\}$ holds.
Define the operator~$A$ as the Laplacian with the dense domain
 $\Dom(A) = \Ran(T) \dot + \Ran(\G)$
  and introduce~$\Gd$  on
  $\Dom(\Gd) = \Dom(A)$ as the trace operator~$\gamma_0$  extended
  as the null mapping to set~$\Ran (T) = \Dom(A)\setminus \Ran(\G)$.
Denote~$\Ad$ the restriction of~$A$ to $\Ran(T)$.
Trace properties of functions from the Sobolev class~$W_2^2(\D)$
 imply that $\Ad $
  is in fact the Dirichlet Laplacian on~$\Dom(\Ad) =
  \Ran (T)$ and $\Ad T = I$.

Let~$\Gn : u \mapsto \left. \frac{\partial u}{\partial n
}\right|_\bD$ be the trace of the outer normal derivative of $u\in
\Dom(A)$ defined on the dense set of sufficiently smooth functions
in the closure of $\D$.
The integral representation
 for $T = \Ad^{-1}$
 and application of the Fubini theorem
 show that  $\Gn T = \G^* : H \to E$,
 as prescribed in (\ref{Eq:GnTandGnG}), see~\cite{Ryz3}.
All components of the operator node~$\mathfrak M$
from~(\ref{Eq:NodeM})
 are now completely determined except for the parameter~$\Lambda$
  defined on domain~$\Dom(\Lambda) \subset L^2(\bD)$.
Below we give three definitions of~$\Lambda$ resulting in
 three boundary value problems for harmonic and analytic functions
 in $\D$.
 We are concerned with the equation~(\ref{Eq:Harmonic2})
 and for simplicity only the case
  $\lambda =0 $ of system~(\ref{Eq:GeneralSystem})
   is discussed.
  Results for the
   spectral problem with any $\lambda \in \Complex $
    easily follow from the abstract considerations of Section~1.
%


\subsection{Poincar\'e problem}
Definitions of operators~$\Gd$ and $\Gn$ given above
 suggest the ``natural'' choice of~$\Lambda$.
 Since $\Gn$ maps a smooth function defined in $\overline\D$
   to the trace of its normal derivative on $\bD$, and
   $\G$ is the operator of harmonic continuation,
   we have for smooth~$\varphi$
\[
   \Gn \G : \varphi \mapsto \left. \frac{\partial u^\varphi_0}{\partial n }\right|_\bD
\]
 where $u^\varphi_0$ is the solution to $A u = 0$
  satisfying boundary condition~$u|_\bD = \varphi$.
Operator~$\Omega := \Gn \G$ is called
 \textit{the Dirichlet-to-Neumann map}
 for the Laplacian~$\Delta$ in $\D$.
It is known that~$\Omega$ defined on the Sobolev class~$W^1_2(\bD)$ is
 selfadjoint in $L^2(\bD)$.
Let~$\Lambda = \Omega$ with the domain
 $\Dom(\Lambda) = W^1_2(\bD)$.
%
%


According to Theorem~\ref{Thm:SolvabilityB0B1},
 for two mappings
  $\bd :  W^{1}_2(\bD) \to L^2(\bD)
 $ and  $\bn : L^2(\bD) \to L^2(\bD)$,
and  $ g\in L^2(\bD)$
 the solvability of system
 \begin{equation}\label{Eq:SystemPoincare}
 A u =0, \qquad
 (\bd \Gd + \bn \Gn ) u = g
 \end{equation}
  is equivalent to the solvability of
 \begin{equation}\label{Eq:DNBoundaryCondition}
 (\bd + \bn \Omega) \varphi = g.
 \end{equation}
Let $\tilde{\bd}, \tilde{\bn}, \tilde\gamma $ be
 complex valued measurable functions on $\bD$.
Define operators~$\bd$ and $\bn$ by
 \begin{equation}\label{Eq:OperatorsBoB1}
\bd : \varphi \mapsto \tilde{\bd}  \frac{d \varphi }{ds} +
\tilde\gamma \varphi \qquad \qquad
 \bn : \varphi \mapsto \tilde{\bn}\varphi\qquad
 \end{equation}
where $\frac{d}{ds}$ is the operator of (generalized)
differentiation in $L^2(\bD)$.
For sufficiently smooth  $\varphi$  the harmonic function~$u^\varphi =
\G \varphi$
 is continuously differentiable in the closure~$\overline \D$ and
  the
 trace of its tangential derivative~$\frac{\partial u^\varphi}{\partial\tau}$  on
 the boundary~$\bD$ satisfies
\[
 \left. \frac{\partial u^\varphi}{\partial\tau}\right|_\bD = \frac{d\varphi}{ds}
 \]
Thus the boundary condition in (\ref{Eq:SystemPoincare})
 is meaningful at least for harmonic functions~$u \in L^2(\D)$ with boundary
  values from $ W_2^1(\bD)$.
Solvability of (\ref{Eq:SystemPoincare}) with the
choice~(\ref{Eq:OperatorsBoB1}) therefore is determined  by
 the solvability of
\[
 \left(  \tilde{\bd}\frac{d}{ds} +
   \tilde{\bn} \Omega
+ \tilde \gamma
   \right) \varphi = g, \qquad g \in L^2(\bD)
\]
for unknown  $\varphi \in W_2^1(\bD)$.
Since $\Omega \varphi = \left. \frac{\partial u^\varphi}{\partial
n}\right|_\bD$, this condition can be rendered as
\begin{equation}\label{Eq:ObliqueDerivative}
 \left. \tilde{\bd}\frac{\partial u}{\partial \tau}\right|_\bD +
 \left. \tilde{\bn} \frac{\partial u}{\partial n}\right|_\bD
 +  \left. \tilde \gamma  u\right|_\bD
= g, \qquad g \in L^2(\bD)
\end{equation}
for the unknown  $u$ harmonic in $\D$.
When $\tilde{\bd}$, $\tilde{\bn}$, $\tilde\gamma$, and $g$ are
 sufficiently regular and
 real
 valued, the problem~(\ref{Eq:ObliqueDerivative})
  reduces to the classical Poincar\'e's problem for harmonic
  functions~\cite{Mus}.
%


%
%

\subsection{Hilbert problem}
Hilbert problem in the domain~$\D$
 consists in seeking an analytic function~$w = u + iv $ defined in~$\D$
  with the real and imaginary parts $u$, $v$ satisfying following condition on the
  boundary~$\bD$
\begin{equation}\label{Eq:ClassicalHilbertProblem}
   a (s) u (s)+ b(s) v (s) = g(s),
\end{equation}
 with real valued functions $a$, $b$, and $g$.
For simplicity we consider the case when~$\D$ is the unit disc~$
\Disc  = \{ z \in \Complex \,\mid\, |z| < 1\}$ in the complex plane
with the boundary ~$\Torus = \{z\in \Complex \, \mid\, |z| = 1\}$.
In order to apply the general schema we need to recall some
properties of Hilbert transform~$\mathcal H$ acting
in~$L^2(\Torus)$, see~\cite{Gar, Koos}.
The operator~$\mathcal H$ is bounded in $L^2(\Torus)$ and
  for real~$\varphi \in L^2(\Torus)$ the
  function~$\varphi  + i \mathcal H  \varphi$ is boundary
  value of the function~$w = u + i v $ analytic in $\Disc$.
In other words, if $w  =u + i v $ is analytic in $\Disc$ with real
valued harmonic functions $u$, $v$ and such that the trace~$\varphi
 = \left. u\right|_\Torus$ is in $L^2(\Torus)$ ,
 then $\widetilde \varphi = \left. v\right|_\Torus$ is also in $L^2(\Torus)$
 and functions $\varphi$,  $\widetilde \varphi$
 are related by equality~$\widetilde \varphi = \mathcal H \varphi$.
The function~$\widetilde \varphi $ is called \textit{the harmonic
conjugate} of $\varphi$.

 Define $\Lambda$ to be the Hilbert transform, $\Lambda = \mathcal H$.
 Then the boundary condition~(\ref{Eq:GeneralSystem})
  results in the equation
$ ( \bd + \bn \mathcal H ) \varphi = g $ that
 can be rewritten as
 \begin{equation}\label{Eq:HilbertProblem}
  \bd \varphi +  \bn \widetilde\varphi = g
 \end{equation}
 Let $\bd : \varphi \mapsto a\varphi$,
 $\bd : \varphi \mapsto b\varphi$
be two multiplication
 operators by measurable functions~$a, b$ on $\Torus$.
Under additional assumption that  $a$, $b$, $\varphi $, and $g$ are
real valued,
 the condition~(\ref{Eq:HilbertProblem}) corresponds to the Hilbert
problem~(\ref{Eq:ClassicalHilbertProblem}) for unknown function~$w =
u + i v $ analytic in $\Disc$.
If $\varphi\in L^2(\Torus)$ solves the equation
 \begin{equation}\label{Eq:HilbertProblemAB}
  a (s) \varphi (s)+ b(s)(\mathcal H \varphi )(s) = g(s),
 \end{equation}
 for almost all $s\in \Torus$
 then the
 solution to (\ref{Eq:ClassicalHilbertProblem})
 is $w = u + i v $ with real
 and imaginary parts $u =  \G\varphi$
 and $v =  \G \mathcal H \varphi$.

In the language of open systems theory the
 equation~(\ref{Eq:HilbertProblem}) can be treated as redefined input of
 the system~$\widehat{\mathfrak M}$ corresponding to the operator
  node~(\ref{Eq:NodeM}) with $\Lambda = \mathcal H $.
As an example, consider the left hand side
of~(\ref{Eq:HilbertProblem})
 with $\bd = 1$, $\bn = i$
 as the input of the new system~$\widehat {\mathfrak N}$
  and with $\bd = 1$, $\bn = -i$ as
 the output of $\widehat {\mathfrak N}$.
 Then the feedthrough operator of $\widehat {\mathfrak N}$
 is the map
 \[
  \Theta : (I + i \mathcal H)\varphi \mapsto (I - i \mathcal H
  )\varphi, \qquad \varphi \in L^2(\Torus)
 \]
which can not be written in the form~$\Theta = (I - i \mathcal H)(I
+ i \mathcal H)^{-1}$ because $I + i \mathcal H$ is not boundedly
invertible.
 Property~$\mathcal H^2 = -I$ of the Hilbert transform
 yields $(I + i \mathcal H)(I - i \mathcal H) = 0$
 and therefore $\Ker(I + i \mathcal H)$ is not trivial.
 In fact, $\Ker(I + i \mathcal H) = \Ran(I - i \mathcal H)$.
 Thus the mapping $\Theta$ is the linear relation on
 the space~$L^2(\Torus)\oplus
 L^2(\Torus)$.
 However, its restriction to the set~$(I + i \mathcal H)\Re( L^2(\Torus))  \oplus \{0\}$
 where $\Re( L^2(\Torus)) $ is the set of all real valued functions from
 $L^2(\Torus)$,
 defines an operator~$\theta = (I - i \mathcal H) (I + i \mathcal H
  )^{-1}$.
 It maps boundary values of functions~$w = u + iv$
 analytic in $\Disc$ with $u|_\Torus \in \Re( L^2(\Torus))$
  to the boundary
 values of complex conjugate function~$\bar w = u - iv$.
 Note that the operator~$\theta$ is
 not linear
 over the field of complex numbers
 because
  $a \theta w \neq \theta aw$,
 for $w \in \Dom(\theta)$ and
 $a\in \Complex$ unless $a$ is a real number.

\subsection{Riemann problem}
The Riemann problem for analytic functions is another case that can
 be studied by means of the general theory of Section~1.
Let $\D$ be the simply connected bounded domain~$\D\subset \Complex$
with regular boundary~$\bD$
and  $B$, $g$ be  measurable
complex valued functions on $\bD$.
A pair of functions $\Phi^\pm$
 is a solution to the corresponding Riemann problem
 if $\Phi^+$ is analytic in $\D$, $\Phi^-$ is analytic
 in~$\Complex \setminus \overline\D$, non-tangential boundary
 values of~$\Phi^\pm$
 on the contour~$\bD$ exist almost everywhere, and
 \begin{equation}\label{Eq:RimannProblemPhiPhi}
 \Phi^+(s) - B(s) \Phi^-(s) = g(s), \qquad \textit{   a.e.  } s \in \bD
 \end{equation}
Note that
 all considerations carried out in the beginning of this section for
  the Laplacian remain fully applicable,
 as $\Phi^\pm$ are linear combinations of harmonic functions
 defined in their corresponding domains.
%


 Let
$\mathcal S$ be the Cauchy singular integral operator on the
contour~$\bD$ defined for $\varphi  \in L^2(\bD)$ by
\[
\mathcal S : \varphi \mapsto \Phi(s) = \frac 1 {\pi i}\int_\bD
\frac{\varphi(t) dt}{t - s } =
 \frac 1 {\pi i}  \lim_{\varepsilon \downarrow 0 }\int_{|t - s | >
 \varepsilon}
\frac{\varphi(t) dt}{t - s }
\]
Operator $\mathcal S$
 is bounded in $L^2(\bD)$.
For notational convenience denote $D^+ = D$ and $D^- = \Complex \setminus \overline\D$.
 Two functions
\[
 \Phi^\pm(z) = \frac 1 {\pi i}\int_\bD \frac{\varphi(t) dt}{t - z }, \qquad z \in D^\pm
\]
where $\varphi  \in L^2(\bD)$
 are analytic and possess non-tangential boundary values
 almost everywhere on the contour~$\bD$
 \[
 \lim_{ z\to s, \  z\in D^\pm } \Phi^\pm(z) = \Phi^\pm(s),\qquad   \textit{   a.e.  } s\in\bD
 \]
The Sokhotzki-Plemelj formulae~\cite{Beg, Ga, Mus}
\begin{equation}\label{Eq:SokhPlem}
 \Phi^+ (s)  =  \varphi(s) + \Phi(s),
 \qquad
 \Phi^- (s)  = -  \varphi(s) + \Phi(s)
 \qquad   \textit{   a.e.  } s\in\bD
\end{equation}
and boundedness of~$\mathcal S$
 show that $\Phi^\pm \in L^2(\bD)$.
%
%
%

Introduce two multiplication operators~$\bd : \varphi \mapsto
a(s)\varphi(s)$, $\bn : \varphi \mapsto b(s) \varphi (s)$, where $a,
b$ are measurable functions of $s\in \bD$.
Then the choice $\Lambda : \varphi \mapsto \mathcal S \varphi $
 and the
 boundary condition from~(\ref{Eq:GeneralSystem}) leads to the
 equation for unknown~$\varphi \in L^2(\bD)$
\begin{equation}\label{Eq:RiemannProblem}
a(s) \varphi(s) + b(s) (\mathcal S\varphi) (s)= g(s)
 \end{equation}
Put~$a =  A + B$, $b = A - B$ with some measurable functions $A, B$
defined on $\bD$.
Then for $\Phi = \mathcal S \varphi$
 the equation~(\ref{Eq:RiemannProblem}) takes the form
\[
 a \varphi + b \mathcal S \varphi  =
  (A + B)\varphi + (A - B)\mathcal S \varphi  =
 A\left(\varphi + \Phi \right) - B \left( -  \varphi + \Phi\right)
 \]
 Therefore due to (\ref{Eq:SokhPlem}) the
 equation~(\ref{Eq:RiemannProblem}) becomes
 %
\[
 A(s) \Phi^+(s) - B(s) \Phi^-(s) = g(s), \qquad s \in \bD
\]
%
%
%
For $A(s) =1$ we arrive at the  Riemann boundary value
problem~(\ref{Eq:RimannProblemPhiPhi}).
%


%

%
Other types of
  boundary value problems can be described by
 the
 equation~(\ref{Eq:RiemannProblem}) if
 we continue to treat $a$ and $b$ as linear operators.
 For example, let~$\tau: \varphi(s) \to \varphi(\alpha(s))$, $s \in \bD$
  be the composition operator where~$\alpha(s)$ is an arbitrary one-to-one mapping
  of the contour~$\bD$  onto itself with
   continuous derivative $\alpha'(s) \neq 0$.
 The choice~$a = A\tau + B$, $b = A\tau - B$ where
 two multiplication operators~$A$, $ B$ are as above,
 results
 in the so-called shifted Riemann boundary value problem with,
 see~\cite{Ga} for details. 
\[
 A(s) \Phi^+ [\alpha(s)] - B(s) \Phi^-(s) = g(s)
\]
Note in conclusion that
 the case $\lambda \neq 0$ of the general
  spectral problem~(\ref{Eq:GeneralSystem})
   appears to be irrelevant
   for the study of analytic functions
    in the paper's context.
 However, the spectral theory approach
  may prove beneficial in the study
   of boundary value problems for the first-order
    differential operators of complex analysis,
     most notably, Cauchy-Riemann and Beltrami
      operators on domains
       (see for example \cite{Beg} for their definitions).

\end{document}